\documentclass[
reprint,
amsmath,amssymb,
aps,
]{revtex4-2}

\usepackage{graphicx}% Include figure files
\usepackage{dcolumn}% Align table columns on decimal point
\usepackage{bm}% bold math
\usepackage{hyperref}% add hypertext capabilities
\usepackage{subcaption}

\usepackage{amsthm}

\usepackage{braket}
\usepackage{caption}
\usepackage{xcolor}
\usepackage{diagbox}
\usepackage{booktabs}
\usepackage{adjustbox}
\usepackage{multirow}
\usepackage[mathlines]{lineno}
\usepackage{ragged2e}
\usepackage{pifont}
\usepackage{makecell}
\usepackage{algorithm}
\usepackage{algpseudocode}
\usepackage{soul}
\usepackage{cleveref}
\usepackage{tikz}
\usepackage{array}

\newtheorem{proposition}{Proposition}
\newtheorem{remark}{Remark}
\newtheorem{lemma}{Lemma}

\newtheorem{theorem}{Theorem}

\newtheorem{problem}{Problem}
\newtheorem{definition}{Definition}
\newtheorem{corollary}{Corollary}

\newcommand{\Tr}{\mathrm{tr}}

\newcommand{\calE}{\mathcal{E}}
\newcommand{\calD}{\mathcal{D}}

\newcommand{\calN}{\mathcal{N}}

\newcommand{\id}{\mathrm{id}}
\newcommand{\cH}{\mathcal{H}}
\newcommand{\cK}{\mathcal{K}}
\newcommand{\cL}{\mathsf{L}}
\newcommand{\cD}{\mathsf{D}}
\newcommand{\C}{\mathbb{C}}
\newcommand{\normone}[1]{\left\|#1\right\|_{1}}
\newcommand{\normtwo}[1]{\left\|#1\right\|_{2}}
\newcommand{\normd}[1]{\left\|#1\right\|_{\diamond}}
\newcommand{\rank}{\operatorname{rank}}

\newcommand{\N}{\mathcal{N}}
\newcommand{\EN}{E_{\mathcal{N}}}

\newcommand{\PT}{\Gamma}

\newcommand{\tnorm}[1]{\left\lVert #1\right\rVert_{1}}

\begin{document}

% \preprint{APS/123-QED}

\title{Refining transposition bounds from traceless quantum perturbations}

% Author block
% \author{Hyunho Cha}
% % \author{Hyunho Cha${}^{1}$}
% \email{aiden132435@cml.snu.ac.kr}
% \author{Jungwoo Lee}
% \email{junglee@snu.ac.kr}
% \affiliation{NextQuantum and Department of Electrical and Computer Engineering, Seoul National University, Seoul 08826, Republic of Korea}

\author{Hyunho Cha}
\email{ovalavo@snu.ac.kr}
\author{Jungwoo Lee}
\email{junglee@snu.ac.kr}
\affiliation{NextQuantum and Department of Electrical and Computer Engineering, Seoul National University, Seoul 08826, Republic of Korea}

\date{\today}% It is always \today, today,
             %  but any date may be explicitly specified

\begin{abstract}
The transpose map is the canonical nonphysical map in quantum information that is central to entanglement tests and to converse bounds for quantum capacity.
It is a longstanding conjecture that transposition satisfies a dimension-independent strict submultiplicativity for physically relevant channel differences.
We answer this affirmatively by proving that for every quantum channel $T$ on $\cL(\mathbb{C}^d)$, $\|\Theta\circ(\id-T)\|_\diamond \le \frac{1}{\sqrt{2}}\|\Theta\|_\diamond \|\id-T\|_\diamond$. The same bound holds for all Hermiticity-preserving trace-annihilating maps.
Because a trace-preserving channel produces a traceless deviation from the identity, transposition cannot act as strongly on such physically generated errors as it can on arbitrary perturbations.
Thus, although transposition is extremal in the unrestricted diamond-norm geometry, physical constraints mandate a universal margin in the submultiplicative bound. As one immediate consequence, the finite-error Holevo--Werner converse extends from decoding errors $\varepsilon<1/2$ to $\varepsilon<1/\sqrt{2}$.
The same tracelessness mechanism also yields rank-sensitive continuity refinements for negativity and logarithmic negativity, demonstrating its potential for broader application in quantum information.
\end{abstract}

%\keywords{Suggested keywords}%Use showkeys class option if keyword
                              %display desired
\maketitle

%\tableofcontents

\section{Introduction}

Quantum capacity is one of the most subtle basic quantities in quantum information theory \cite{lloyd1997capacity, divincenzo1998quantum, shor2004equivalence, devetak2005private, devetak2005capacity, devetak2008resource, hayden2008decoupling, wilde2013quantum, cubitt2015unbounded}.
Even its coding theorem is intrinsically regularized \cite{devetak2005capacity, smith2010quantum, singh2022coherent}, additivity fails in dramatic ways \cite{smith2010quantum, singh2022coherent, smith2025additivity}, and superactivation shows that two zero-capacity channels can combine to transmit quantum information \cite{oppenheim2008quantum, smith2008quantum}.
These features make simple, computable converse bounds especially valuable, particularly when one wants statements that remain meaningful at nonzero decoding error rather than only in the asymptotic vanishing-error limit \cite{holevo2001evaluating, matthews2014finite, takeoka2014squashed, leung2015power, tomamichel2016quantum, tomamichel2016strong, wang2016semidefinite, christandl2017relative, berta2018amortization, wang2018semidefinite}.

Among the mathematical objects that enter such converse bounds, the transpose map holds a unique importance \cite{holevo2001evaluating, rains2001semidefinite, horodecki2002method, muller2016positivity}. It is the canonical example of a positive map that is not completely positive, it lies behind the Peres-Horodecki separability criterion, and it underpins entanglement measures such as the negativity \cite{stinespring1955positive, choi1975completely, peres1996separability, lee2000partial, horodecki2001separability, vidal2002computable, plenio2005logarithmic, toth2018quantum}. At the same time, the diamond norm provides the natural stabilized notion of distance for quantum channels and linear maps acting on quantum states \cite{aharonov1998quantum, acin2001statistical, kitaev2002classical, gilchrist2005distance, johnston2009computing, watrous2009semidefinite, watrous2018theory}.
The transpose map is extremal in this geometry because its diamond norm grows linearly with dimension, reflecting its maximal deviation from the behavior of completely positive maps \cite{tomiyama1983transpose, paulsen2002completely, watrous2018theory}.
This makes the transpose map both mathematically rigid and operationally significant, and helps explain why it repeatedly appears at the interface of entanglement theory, operator theory, and quantum Shannon theory \cite{holevo2001evaluating, fiuravsek2002structural, horodecki2002method, paulsen2002completely, pisier2003introduction, korbicz2008structural, lim2011experimental, kalev2013optimal, muller2016positivity, regula2021fundamental, regula2021operational}.

Our motivating example is the transposition-based upper bound on quantum capacity introduced by Holevo and Werner \cite{holevo2001evaluating}. That bound remains valid in the presence of finite decoding error and has the virtue of being explicit and conceptually transparent. Since then, stronger converse bounds have been developed through semidefinite programming \cite{wang2016semidefinite, wang2017semidefinite}, PPT-assisted methods \cite{muller2016positivity, wang2016semidefinite, wang2018semidefinite}, and Rains-type quantities \cite{tomamichel2016strong, berta2018amortization}. Nevertheless, the Holevo--Werner transposition bound retains a special status. It is arguably the cleanest converse in which a genuinely non-completely-positive map enters in an essential way, and it therefore provides a testing ground for understanding what transposition really contributes to quantum coding limits.

In this work, we revisit the precise step at which the usual transposition argument loses sharpness. Standard proofs rely on the generic submultiplicativity of the diamond norm, applied to the composition of transposition with the deviation of a channel from the identity \cite{paulsen2002completely, pisier2003introduction, johnston2009computing, watrous2009semidefinite}. However, that deviation is far from arbitrary. If $T$ is a quantum channel, then the map $\id - T$ is automatically Hermiticity preserving and trace-annihilating. After adjoining an ancilla, it therefore sends every bipartite state to a traceless Hermitian operator. This extra structure is invisible to the usual worst-case norm inequality. It has thus raised a natural open question \cite{iqoqi_problem37}: when the perturbation is restricted to the form $\id-T$ for a quantum channel $T$, is the impact of transposition strictly smaller than the general diamond norm bound would predict?

We answer this question in the affirmative. Our main theorem establishes a \emph{dimension-independent strict submultiplicativity} statement for the transposition map on channel differences. A key feature of our proof is the temporary passage to the Hilbert--Schmidt norm, where transposition acts as an isometry \cite{bhatia2013matrix}. The argument isolates the structural fact that, while transposition is maximally large in diamond norm in the unrestricted worst case, it is uniformly better behaved on trace-annihilating Hermiticity-preserving perturbations.

The result identifies a new rigidity property of the transpose map in stabilized norm geometry. The theorem also extends immediately beyond channel differences to arbitrary Hermiticity-preserving, trace-annihilating linear maps, which suggests that the phenomenon is not accidental and may reflect a broader principle about how nonphysical maps interact with perturbations that arise from physical constraints.
Our negativity and logarithmic-negativity consequences provide instances of this broader reach.
At the same time, we show that the universal constant is not yet fully understood, framing our result as both a definitive contribution and the basis for a well-posed open problem.

The information-theoretic payoff is immediate. In the finite-error Holevo--Werner converse, the familiar restriction to errors below $1/2$ enters exactly through the generic submultiplicativity property. Replacing that step by the strict bound proved here enlarges the admissible error regime to below $1/\sqrt{2}$ in the diamond-error convention $\frac{1}{2}\|\Phi-\Psi\|_\diamond$.
Beyond this coding application, the same mechanism also gives rank-sensitive continuity improvements for negativity and logarithmic negativity, where the relevant state differences are again Hermitian and traceless.

\section{Setup and definitions}

Let $\cH \cong \mathbb{C}^d$ be a $d$-dimensional Hilbert space.
We write $\cL(\cH)$ for the space of linear operators on $\cH$, and
\[
\cD(\cH)=\{\rho\in \cL(\cH): \rho\succeq 0,\ \Tr(\rho)=1\}
\]
for the set of density operators.

\paragraph{Transposition map.}
Fix an orthonormal basis $\{|1\rangle,\dots,|d\rangle\}$ of $\cH$. The transposition map
$\Theta:\cL(\cH)\to\cL(\cH)$ is defined by
\[
\Theta(X) = X^{\mathsf T},
\]
which depends on the chosen basis.

\paragraph{Quantum channels.}
A \emph{quantum channel} $T:\cL(\cH)\to\cL(\cH)$ is a completely positive trace-preserving (CPTP) linear map.

\paragraph{Norms.}
The trace norm and Hilbert--Schmidt norm of an operator $X$ are
\[
\normone{X}=\Tr\sqrt{X^\dagger X},\qquad \normtwo{X}=\sqrt{\Tr(X^\dagger X)}.
\]
The \emph{diamond norm} of a linear map $\Phi:\cL(\cH)\to\cL(\cH)$ is
\[
\normd{\Phi}\;:=\;\sup_{k\ge 1}\;\sup_{\rho\in \cD(\cH\otimes\mathbb{C}^k)}\;
\normone{(\Phi\otimes \id_k)(\rho)}.
\]
It is standard that for maps acting on $\cL(\mathbb{C}^d)$, the supremum over $k$ can be restricted to $k=d$. We will use this without further comment and write $\id$ for $\id_d$ throughout the paper.

\section{Stronger submultiplicativity for the diamond norm}

The first main result of this work is given by the following theorem.

\begin{theorem}
\label{thm:main}
Let $\cH\cong \mathbb{C}^d$ and let $T:\cL(\cH)\to\cL(\cH)$ be a quantum channel. Then
\begin{equation}
\label{equation:main_theorem_eq}
\normd{\Theta\circ(\id-T)} \;\le\; \alpha \normd{\Theta} \normd{\id-T}
\end{equation}
with $\alpha=1/\sqrt{2}$. In particular, since $\normd{\Theta}=d$~\emph{\cite{tomiyama1983transpose}},
\[
\normd{\Theta\circ(\id-T)} \;\le\; \frac{d}{\sqrt2} \normd{\id-T}.
\]
\end{theorem}

The proof proceeds by reducing to a matrix inequality for the partial transpose applied to traceless Hermitian matrices.
The specific form $\Theta \circ (\id-T)$ appears in~\cite{holevo2001evaluating}, but the analysis therein relied on the standard diamond-norm submultiplicativity corresponding to the case $\alpha=1$.
The possibility of obtaining a \emph{dimension-independent} constant $\alpha < 1$ for a stronger submultiplicativity result remained open.

We first present a few useful lemmas.

\begin{lemma}
\label{lem:traceless-hs}
If $X=X^\dagger$ and $\Tr(X)=0$, then
\[
\normtwo{X} \le \frac{1}{\sqrt2} \normone{X}.
\]
\end{lemma}

\begin{proof}
Let $\{\lambda_i\}_{i=1}^n$ be the eigenvalues of $X$ (real because $X$ is Hermitian).
Write the positive eigenvalues as $p_1,\dots,p_r\ge 0$ and the negative eigenvalues as
$-q_1,\dots,-q_s$ with $q_1,\dots,q_s > 0$.
Since $\Tr(X)=\sum_i\lambda_i=0$, we have
\[
\sum_{i=1}^r p_i \;=\; \sum_{j=1}^s q_j \;=:\; t.
\]
Hence
\[
\normone{X}=\sum_i|\lambda_i|=\sum_{i=1}^r p_i+\sum_{j=1}^s q_j = 2t.
\]
Also,
\begin{align}
\label{eq:cross_terms_ineq}
\normtwo{X}^2 & = \sum_i\lambda_i^2=\sum_{i=1}^r p_i^2+\sum_{j=1}^s q_j^2\nonumber\\
& \le \Big(\sum_{i=1}^r p_i\Big)^2+\Big(\sum_{j=1}^s q_j\Big)^2
= t^2+t^2=2t^2,
\end{align}
where we used the elementary inequality $\sum a_i^2\le (\sum a_i)^2$ for nonnegative $a_i$.
Therefore $\normtwo{X}\le \sqrt{2}\,t=\normone{X}/\sqrt{2}$.
\end{proof}

\begin{lemma}
\label{lem:main_lemma}
Let $Y$ be an operator on an $N$-dimensional Hilbert space with rank $r$. Then
\[
\normone{Y} \le \sqrt{r}\,\normtwo{Y} \le \sqrt{N}\,\normtwo{Y}.
\]
\end{lemma}

\begin{proof}
Let $s_1,\dots,s_r$ be the nonzero singular values of $Y$.
Then $\normone{Y}=\sum_i s_i$ and $\normtwo{Y}=\sqrt{\sum_i s_i^2}$.
By Cauchy--Schwarz,
\begin{equation}
\label{eq:singular_values_CS}
\Big(\sum_{i=1}^r s_i\Big)^2 \le r \sum_{i=1}^r s_i^2 \le N \sum_{i=1}^r s_i^2,
\end{equation}
which gives the claim.
\end{proof}

\begin{lemma}
\label{lem:pt-hs}
Let $\cH,\cK$ be finite-dimensional and let $\Theta$ denote transposition on $\cH$.
Then for all $X\in \cL(\cH\otimes\cK)$,
\[
\normtwo{(\Theta\otimes \id)(X)}=\normtwo{X}.
\]
\end{lemma}

\begin{proof}
In the computational basis on $\cH$, the transposition map is an isometry for the Hilbert--Schmidt inner product:
\[
\langle A,B\rangle_\mathrm{HS}:=\Tr(A^\dagger B),\qquad \Tr\big((A^{\mathsf T})^\dagger B^{\mathsf T}\big)=\Tr(A^\dagger B).
\]
Applying this entrywise on $\cH$ and trivially on $\cK$ yields
\[
\Tr\Big(\big((\Theta\otimes\id)(X)\big)^\dagger (\Theta\otimes\id)(X)\Big)=\Tr(X^\dagger X),
\]
i.e.\ $\normtwo{(\Theta\otimes\id)(X)}=\normtwo{X}$.
\end{proof}

% \subsection{Proof of Theorem~\ref{thm:main}}
We are now ready to prove Theorem~\ref{thm:main}.

\begin{proof}[Proof of Theorem~\ref{thm:main}]
Fix an ancilla $\cK\cong \mathbb{C}^d$. For any $\rho\in \cD(\cH\otimes\cK)$,
\[
\Tr\!\Big(\big((\id-T)\otimes\id\big)(\rho)\Big)=\Tr\!\big((\id-T)(\Tr_{\cK}\rho)\big)=0,
\]
since $\id-T$ is trace-annihilating. Therefore,
\begin{align*}
& \left\|\big( (\Theta\circ (\id-T)) \otimes \id\big)(\rho)\right\|_1\\
& \leq_{(\text{Lem.~\ref{lem:main_lemma}})} d \left\|\big((\Theta\circ (\id-T))\otimes \id\big)(\rho)\right\|_2\\
& =_{(\text{Lem.~\ref{lem:pt-hs}})} d \left\|\big( (\id-T)\otimes \id \big) (\rho) \right\|_2\\
& \leq_{(\text{Lem.~\ref{lem:traceless-hs}})} \frac{d}{\sqrt{2}} \left\|\big( (\id-T)\otimes \id \big) (\rho) \right\|_1\\
& \leq \frac{d}{\sqrt{2}}\|\id-T\|_\diamond\\
& =\frac{1}{\sqrt{2}}\|\Theta\|_\diamond\|\id-T\|_\diamond.
\end{align*}
Taking the supremum over $\rho$,
\[
\normd{\Theta\circ(\id-T)} \le \frac{1}{\sqrt2} \normd{\Theta} \normd{\id-T}.
\]
\end{proof}

\begin{remark}
The constant $1/\sqrt{2}$ is universal (independent of $d$ and $T$) and arises from the
fact that $\big((\id-T)\otimes\id\big)(\rho)$ is always traceless because $\id-T$ is trace-annihilating.
Moreover, Theorem~\ref{thm:main} immediately generalizes by replacing $\id-T$ with any linear map
$\Psi : \cL(\cH) \to \cL(\cH)$ that is Hermiticity-preserving and trace-annihilating. In that case the same argument yields
$
\|\Theta \circ \Psi\|_\diamond
\le \frac{1}{\sqrt{2}} \|\Theta\|_\diamond \|\Psi\|_\diamond.
$
\end{remark}

\subsection{Improved finite-error bound for quantum capacity}

Holevo and Werner introduced a simple upper bound on quantum capacity in terms of the diamond norm of the transposed channel when a fixed nonzero decoding error is permitted \cite{holevo2001evaluating}.
Here we adopt the common convention
\[
\varepsilon := \frac{1}{2}\left\| \id - \mathcal D \circ T^{\otimes m}\circ \mathcal E \right\|_{\diamond},
\]
so that the Holevo--Werner bound is valid for $\varepsilon<\tfrac12$.

The place where this restriction enters the argument is through the inequality $\|\Theta\circ \Psi\|_{\diamond}\le \|\Theta\|_{\diamond}\,\|\Psi\|_{\diamond}$ applied to
the map $\Psi := \id - \mathcal D \circ T^{\otimes m}\circ \mathcal E$.
However, $\Psi$ is always Hermiticity-preserving and trace-annihilating.
Therefore Theorem~\ref{thm:main} applies and yields the strictly better bound
$\|\Theta\circ \Psi\|_{\diamond}\le \frac{1}{\sqrt2}\|\Theta\|_{\diamond}\,\|\Psi\|_{\diamond}$.
Plugging this refinement into the proof extends the finite-error regime from
$\varepsilon<\frac{1}{2}$ to $\varepsilon<\frac{1}{\sqrt{2}}$.

\begin{corollary}[partial transposition bound for $\varepsilon$-quantum capacity]
Let $T$ be a channel. If there exists a coding scheme $(\calE,\calD)$ using $m$ uses of $T$ to transmit a message of dimension $d$ with error
$$
\varepsilon = \frac{1}{2} \left\| \id - \calD \circ T^{\otimes m} \circ \calE \right\|_{\diamond} < \frac{1}{\sqrt{2}},
$$
then
\[
\frac{\log_2 d}{m} \le \log_2 \| \Theta \circ T \|_{\diamond} + \frac{1}{m} \log_2 \!\left( \frac{1}{1 - \sqrt{2} \varepsilon} \right),
\]
improving upon the finite-blocklength converse of \cite{holevo2001evaluating}.
Moreover, taking $m \to \infty$ yields the asymptotic statement
\[
Q_{\varepsilon}(T) \le \log_2 \| \Theta \circ T \|_{\diamond} \quad \text{for all } \varepsilon < \frac{1}{\sqrt{2}}.
\]
\end{corollary}

More generally, any further improvement of the universal constant in Theorem~\ref{thm:main} would translate
directly into a larger admissible error range in this corollary, i.e., $\varepsilon < \frac{1}{2\alpha}$.

\section{Negativity in open system entanglement dynamics}

In \cite{tiersch2013universality}, the Lipschitz constant is derived for the continuity of negativity.

\begin{definition}[Negativity]
For a bipartite state $\rho$ acting on a Hilbert space $\mathcal{H}_A \otimes \mathcal{H}_B$, the \emph{negativity} $\calN(\rho)$ is defined as
\begin{equation}
\calN(\rho) := \frac{\tnorm{\rho^{T_B}} - 1}{2}.
\end{equation}
\end{definition}

The Lipschitz constant is calculated as

\begin{equation}
\label{eq:previous_neg_ineq}
|\N(\rho)-\N(\omega)|\le \frac12\tnorm{(\rho-\omega)^{T_B}}
\le \frac{d'}{2}\tnorm{\rho-\omega},
\end{equation}
where $d'$ is the dimension of the smaller subsystem.

We can derive a rank-sensitive sharpening of this inequality, which becomes a strict improvement for balanced bipartitions.

\begin{proposition}[Continuity of negativity]
\label{prop:tiersch}
Let $\rho,\omega$ be bipartite states on a total $d$-dimensional Hilbert space $\mathcal{H}_A \otimes \mathcal{H}_B$, let $\PT=T_B$, and set
\begin{equation}
r=\rank\!\left(\rho^\PT-\omega^\PT\right).
\end{equation}
Then
\begin{equation}
|\N(\rho)-\N(\omega)|
\le \frac12\sqrt{\frac r2}\,\tnorm{\rho-\omega}
\le \frac12\sqrt{\frac d2}\,\tnorm{\rho-\omega}.
\label{eq:negativity-strict}
\end{equation}
Combining Eqs.~\eqref{eq:previous_neg_ineq} and \eqref{eq:negativity-strict} gives
\begin{equation}
|\N(\rho)-\N(\omega)|
\le \frac12\min\left\{d',\sqrt{\frac r2}\right\}
\tnorm{\rho-\omega}.
\label{eq:negativity-min}
\end{equation}
This is strictly sharper than their Eq.~\eqref{eq:previous_neg_ineq} whenever $r<2d'^2$.  In particular, if $d_A=d_B$, the Lipschitz constant improves by a factor of $\sqrt2$.
\end{proposition}

\begin{proof}
It was shown in \cite{tiersch2013universality} that
\begin{equation}
|\N(\rho)-\N(\omega)|
\le \frac12\tnorm{\rho^\PT-\omega^\PT}.
\end{equation}
The difference $\rho-\omega$ is Hermitian and traceless, so Lemma~\ref{lem:main_lemma} proves Eq.~\eqref{eq:negativity-strict}. If $r<2d'^2$, then $\sqrt{r/2}<d'$. If $d_A=d_B$, then $d=d'^2$ and $\sqrt{r/2}\le\sqrt{d/2}=d'/\sqrt2$.
\end{proof}

\section{Catalyst of finite dimension}

A \emph{quantum catalyst} is an ancillary quantum system that enables certain transformations that would otherwise be impossible. Datta \emph{et al.} use the continuity of \emph{logarithmic negativity} in the analysis of finite-dimensional catalysts \cite{datta2024entanglement}.

\begin{definition}[Logarithmic Negativity]
For a bipartite state $\rho$ acting on a Hilbert space $\mathcal{H}_A \otimes \mathcal{H}_B$, the \emph{logarithmic negativity} $\EN(\rho)$ is defined as
\begin{equation}
\EN(\rho) := \log_2 \tnorm{\rho^{T_A}}.
\end{equation}
\end{definition}

We provide strict refinements of the bounds established in \cite{datta2024entanglement}. Proofs are omitted, as the result below follows directly from Lemma~\ref{lem:main_lemma} applied at the appropriate steps.

\begin{corollary}[Continuity of logarithmic negativity]
\label{corr:datta}
Let $\rho$ and $\sigma$ be bipartite states on a total $d$-dimensional Hilbert space $\mathcal{H}_A \otimes \mathcal{H}_B$, let $\PT=T_A$, and set
\begin{equation}
\epsilon \ge \tnorm{\rho-\sigma},\qquad
r=\rank\!\left(\rho^{\PT}-\sigma^{\PT}\right).
\end{equation}
Then
$$
\left|\tnorm{\rho^{\PT}}-\tnorm{\sigma^{\PT}}\right|
\le \sqrt{\frac r2}\,\epsilon
\le \sqrt{\frac d2}\,\epsilon
$$
and
$$
\left|\EN(\rho)-\EN(\sigma)\right|
\le \frac{1}{\ln2}\sqrt{\frac r2}\,\epsilon
\le \frac{1}{\ln2}\sqrt{\frac d2}\,\epsilon.
$$
\end{corollary}

% \subsection{Positive gap in finite dimension}
% \label{sec:positive_gap}

\section{Positive gap in finite dimension}

Regrettably, the inequality in Theorem~\ref{thm:main} is still not tight for any finite $d$, except in the trivial case $T=\id$, where both LHS and RHS vanish.
This may explain why such a simple proof eluded discovery, as the underlying intuition is less immediate without an equality condition.

Fix any nonzero channel difference $\Phi = \id - T \ne 0$ (equivalently $T \ne \id$). Rewrite the pointwise inequality (for every $\rho$):
\begin{equation}
\label{eq:pointwise_inequality_rewritten}
\underbrace{\normone{(\Theta\circ\Phi\otimes\id)(\rho)}}_{=: L(\rho)} \le \frac{d}{\sqrt2} \underbrace{\normone{(\Phi\otimes\id)(\rho)}}_{=: R(\rho)}.
\end{equation}
Taking suprema gives $\sup_\rho L(\rho) \le \frac{d}{\sqrt{2}} \sup_\rho R(\rho)$. In finite dimension, the suprema are maxima, so equality
$$
\sup_\rho L(\rho) = \frac{d}{\sqrt{2}} \sup_\rho R(\rho)
$$
can only happen if there exists some $\rho_\star$ such that $R(\rho_\star)=\sup_\rho R(\rho)$ and the pointwise inequality Eq.~\eqref{eq:pointwise_inequality_rewritten} is tight at $\rho_\star$, i.e., $L(\rho_\star) = \frac{d}{\sqrt{2}} R(\rho_\star)$. Let $X_\star := (\Phi \otimes \id)(\rho_\star)$. Since $\Phi \ne 0$ and $\rho_\star$ maximizes $R$, we have $X_\star \ne 0$. Now, the chain of inequalities shows
$$
\normone{(\Theta\otimes\id)(X_\star)} \le d\normtwo{X_\star} \le \frac{d}{\sqrt2} \normone{X_\star}.
$$
So equality in Theorem~\ref{thm:main} forces equality in Lemmas~\ref{lem:traceless-hs}~and~\ref{lem:main_lemma} at some nonzero $X_\star$. But the equality conditions of Lemmas~\ref{lem:traceless-hs}~and~\ref{lem:main_lemma} are incompatible (unless $X=0$).

\paragraph{Equality condition for Lemma~\ref{lem:traceless-hs}.}
For a nonzero traceless Hermitian $X$, equality in Lemma~\ref{lem:traceless-hs} holds \emph{iff} $X$ has exactly two nonzero eigenvalues, $+t$ and $-t$. This follows from Eq.~\eqref{eq:cross_terms_ineq}, where equality requires all cross terms to vanish. Equivalently,
\begin{equation}
\label{eq:contradiction_first_component}
\operatorname{rank}(X_\star)=2 \qquad \text{(for nonzero equality cases)}.
\end{equation}

\paragraph{Equality condition for Lemma~\ref{lem:main_lemma}.}
Equality in Lemma~\ref{lem:main_lemma} holds \emph{iff} $s$ in Eq.~\eqref{eq:singular_values_CS} is proportional to $(1,\dots,1)$, i.e., all singular values are equal. If $Y \ne 0$, that forces all $N$ singular values to be the same positive number, hence
\begin{equation}
\label{eq:contradiction_second_component}
\operatorname{rank}\!\big((\Theta\otimes\id)(X_\star)\big) = \operatorname{rank}(Y_\star) = N = d^2.
\end{equation}

To proceed, we need the following definitions and identities.

\begin{definition}[Column-major vectorization]
For a matrix $A \in \C^{d\times d}$ with entries $A_{ij} = \langle i|A|j\rangle$, define
$$
|\!\operatorname{vec}(A)\rangle := \sum_{i=1}^d \sum_{j=1}^d A_{ij} |i\rangle_\cH \otimes |j\rangle_\cK.
$$
\end{definition}

\begin{definition}
\label{equation:swap_definition}
The swap operator $F \in \cL(\cH \otimes \cK)$ is defined by:
$$
F(|i\rangle \otimes |j\rangle) = |j\rangle \otimes |i\rangle.
$$
\end{definition}

\begin{lemma}[\cite{bruss2000approximate}]
\label{lemma:partial_trace_identical}
For any operator $Z \in \cL(\cH \otimes \cK)$,
$\Tr_\cH\!\big( (T \otimes \id)(Z) \big) = \Tr_\cH(Z)$.
\end{lemma}

\begin{corollary}
\label{corollary:partial_trace_is_zero}
$\Tr_\cH\!\big( (\Phi\otimes\id)(\rho) \big)=0$.
\end{corollary}

\begin{proof}
From Lemma~\ref{lemma:partial_trace_identical} we have
\begin{align*}
\Tr_\cH\!\big((\Phi\otimes \id)(\rho)\big) & = \Tr_\cH\!\big(\rho - (T \otimes \id)(\rho)\big)\\
& = \Tr_\cH(\rho) - \Tr_\cH(\rho) = 0.
\end{align*}
\end{proof}

\begin{lemma}[\cite{henderson1981vec, mele2024introduction}]
\label{lemma:mnvec_man}
For any $M \in \C^{d\times d}$ on $\cH$ and $N \in \C^{d\times d}$ on $\cK$,
$$
(M \otimes N) |\!\operatorname{vec}(A)\rangle = |\!\operatorname{vec}(MAN^\top)\rangle.
$$
\end{lemma}

\begin{lemma}
\label{lemma:swap_positions_operators}
For any $X,Y \in \C^{d \times d}$,
$$
F(X\otimes Y) = (Y\otimes X)F.
$$
\end{lemma}

\begin{lemma}
\label{lemma:swap_lemma}
$(\Theta \otimes \id)\big( |\!\operatorname{vec}(A)\rangle\langle\operatorname{vec}(B)| \big) = (\id \otimes A^\top) F (\id \otimes B^\ast)$.
\end{lemma}

\begin{proof}
Set
$
|\Omega\rangle := |\!\operatorname{vec}(\id)\rangle
$
to be the unnormalized maximally entangled vector. Applying Lemma~\ref{lemma:mnvec_man}, we obtain
\begin{align*}
(\id \otimes A^\top)|\Omega\rangle & = |\!\operatorname{vec}(A)\rangle,\\
\langle\Omega|(\id \otimes B^\ast) & = \langle\operatorname{vec}(B)|.
\end{align*}
Therefore
\begin{align*}
& (\Theta \otimes \id) (|\!\operatorname{vec}(A)\rangle\langle\operatorname{vec}(B)|)\\
& = (\Theta \otimes \id) \big( (\id \otimes A^\top)|\Omega\rangle \langle\Omega|(\id \otimes B^\ast) \big)\\
& = (\id \otimes A^\top) (\Theta \otimes \id)(|\Omega\rangle \langle\Omega|) (\id \otimes B^\ast)\\
& = (\id \otimes A^\top) F (\id \otimes B^\ast).
\end{align*}
\end{proof}

Since $X_\star$ is Hermitian of rank 2 and traceless, we can write its spectral decomposition as
$$
X_\star = t\big(|\psi\rangle\langle\psi| - |\phi\rangle\langle\phi|\big),
$$
for some orthogonal $|\psi\rangle, |\phi\rangle \in \cH \otimes \cK$ and some $t>0$. From Corollary~\ref{corollary:partial_trace_is_zero} we have $\Tr_\cH(X_\star)=0$, which implies
$$
\Tr_\cH |\psi\rangle\langle\psi| = \Tr_\cH |\phi\rangle\langle\phi|.
$$
By Uhlmann's theorem \cite{uhlmann1976transition}, there exists a unitary $U$ on $\cH$ such that
$$
|\phi\rangle = (U \otimes \id) |\psi\rangle.
$$
Now choose $A$ such that
$$
|\psi\rangle = |\!\operatorname{vec}(A)\rangle.
$$
Then, using Lemma~\ref{lemma:mnvec_man},
$$
|\phi\rangle = (U \otimes \id) |\!\operatorname{vec}(A)\rangle = |\!\operatorname{vec}(UA)\rangle.
$$
By linearity,
\[
Y_\star = t \big( (\Theta \otimes \text{id})(|\psi\rangle\langle\psi|) - (\Theta \otimes \text{id})(|\phi\rangle\langle\phi|) \big).
\]
Using Lemma~\ref{lemma:swap_lemma} we get
\begin{equation*}
(\Theta \otimes \text{id})(|\psi\rangle\langle\psi|) = (\id \otimes A^\top) F (\id \otimes A^\ast)
\end{equation*}
and
\begin{align*}
& (\Theta \otimes \text{id})(|\phi\rangle\langle\phi|)\\
& = \big(\id \otimes (UA)^\top\big) F \big(\id \otimes (UA)^*\big)\\
& = (\id \otimes A^\top U^\top) F (\id \otimes U^* A^*)\\
& = (\id \otimes A^\top)(\id \otimes U^\top) F (\id \otimes U^*)(\id \otimes A^*).
\end{align*}
Therefore,
\begin{align*}
Y_\star & = t (\id \otimes A^\top) \left[ F - (\id \otimes U^\top) F (\id \otimes U^*) \right] (\id \otimes A^*)\\
& = t (\id \otimes A^\top) F (\id - U^\top \otimes U^*) (\id \otimes A^*),
\end{align*}
where the second equality follows from Lemma~\ref{lemma:swap_positions_operators}.
From submultiplicativity of rank under multiplication,
\[
\text{rank}(Y_\star) \leq \text{rank}(\id - U^\top \otimes U^*).
\]
Now diagonalize $U$ as $U = VDV^\dagger$ with $D = \text{diag}(\lambda_1, \dots, \lambda_d)$, $|\lambda_i| = 1$. Then
\[
U^\top = V^* D V^\top, \qquad U^* = V^* D^* V^\top,
\]
so
\[
U^\top \otimes U^* = (V^* \otimes V^*) (D \otimes D^*) (V^\top \otimes V^\top).
\]
Thus $U^\top \otimes U^*$ is similar to $D \otimes D^*$, whose eigenvalues are $\lambda_i \lambda_j^*$. In particular, for every $i$,
\[
\lambda_i \lambda_i^* = 1,
\]
so
\begin{align*}
& \dim \ker (\id - U^\top \otimes U^*) \geq d\\
\Rightarrow \quad & \text{rank}(\id - U^\top \otimes U^*) \leq d^2 - d.
\end{align*}
Hence
\[
\text{rank}(Y_\star) \leq d^2 - d < d^2,
\]
which contradicts Eq.~\eqref{eq:contradiction_second_component}.

\section{Discussion}

Although semidefinite-programming converses can be numerically stronger, the present refinement remains meaningful for operational, entanglement-theoretic, and conceptual reasons. Operationally, it strengthens the finite-error Holevo--Werner converse by enlarging the admissible decoding error regime, potentially leading to deeper insights into quantum channel capacities. Entanglement-theoretically, the same traceless-Hermitian mechanism yields continuity refinements for negativity and logarithmic negativity. Conceptually, it shows that the former thresholds were byproducts of a coarse use of operator-norm inequalities. Once the perturbation is restricted to the physically natural class of Hermiticity-preserving trace-annihilating maps, transposition exhibits a genuine universal margin.
This is a concrete instance in which physical admissibility shrinks the effective action of a nonphysical map.
We expect further refinements of this analysis to be useful elsewhere in quantum information whenever physical constraints enforce tracelessness.

The remaining gap is now narrower and concrete. Our finite-dimensional analysis shows that the constant $\alpha=1/\sqrt{2}$ is not tight for any nonzero channel difference, whereas the lower-bound construction implies that no dimension-independent constant below $2/\pi \approx 0.64$ can hold uniformly. Determining the optimal constant would therefore complete the geometric picture and identify the exact finite-error threshold certified by this transposition-based converse.

\begin{problem}
Determine the optimal absolute constant $\alpha_\ast$ in Theorem~\ref{thm:main}. At present,
$
\frac{2}{\pi} \le \alpha_\ast \le \frac{1}{\sqrt{2}}.
$
In particular, does one have the strict improvement $\alpha_\ast < \frac{1}{\sqrt{2}}$?
\end{problem}

\section*{Acknowledgments}

The authors are grateful to Miguel Navascu\'es for suggesting a streamlined derivation.

This work is in part supported by the National Research Foundation of Korea (NRF, RS-2024-00451435 (20\%), RS-2024-00413957 (20\%)), Institute of Information \& communications Technology Planning \& Evaluation (IITP, RS-2025-02305453 (15\%), RS-2025-02273157 (15\%), RS-2025-25442149 (15\%), RS-2021-II211343 (15\%)) grant funded by the Ministry of Science and ICT (MSIT), Institute of New Media and Communications (INMAC), and the BK21 FOUR program of the Education, Artificial Intelligence Graduate School Program (Seoul National University), and Research Program for Future ICT Pioneers, Seoul National University in 2026.

% The \nocite command causes all entries in a bibliography to be printed out
% whether or not they are actually referenced in the text. This is appropriate
% for the sample file to show the different styles of references, but authors
% most likely will not want to use it.
% \nocite{*}
% \bibliography{apssampv1}% Produces the bibliography via BibTeX.
% \putbib[apssampv1]

% \section{A lower bound on $\alpha$}
% \label{sec:alpha_lower_bound}

\appendix

\section{A lower bound on $\alpha$}

Unofficially, an analytic lower bound of $\alpha > 0.63$ was previously established through unitary channels. For completeness, we provide a self-contained reproduction of that result.

Fix an integer $d\ge 2$, and consider the unitary
\[
U_d=\operatorname{diag}(1,\omega,\omega^2,\dots,\omega^{d-1}),
\qquad
\omega=e^{2\pi i/d}.
\]
Let
\[
\operatorname{Ad}_{U_d}(X)=U_d X U_d^\dagger
\]
be the corresponding unitary channel, and define
\[
\Lambda_d:=\Theta\circ(\id - \operatorname{Ad}_{U_d}).
\]
The following two statements hold:
\begin{equation}
\label{eq:Lambda-lower}
\|\Lambda_d\|_{\diamond}
\ge
2\cot\!\left(\frac{\pi}{2d}\right)
\end{equation}
and
\begin{equation}
\label{eq:id-minus-unitary-exact}
\|\mathrm{id}-\operatorname{Ad}_{U_d}\|_{\diamond}=2.
\end{equation}
Once these are established, we can substitute $T=\operatorname{Ad}_{U_d}$ into Eq.~\eqref{equation:main_theorem_eq}. From Eqs.~\eqref{eq:Lambda-lower} and \eqref{eq:id-minus-unitary-exact}, we obtain
\[
2\cot\!\left(\frac{\pi}{2d}\right)
\le
\|\Lambda_d\|_{\diamond}
\le
\alpha \|\Theta\|_{\diamond} \|\id - \operatorname{Ad}_{U_d}\|_{\diamond}
=
2\alpha d.
\]
Dividing both sides by $2d$ and taking the limit $d\to\infty$ implies that any dimension-independent constant $\alpha$ satisfying Eq.~\eqref{equation:main_theorem_eq} for every quantum channel $T$ must satisfy
\[
\alpha \ge \frac{2}{\pi} \approx 0.64.
\]
It remains to establish Eqs.~\eqref{eq:Lambda-lower} and \eqref{eq:id-minus-unitary-exact}.

% \subsection{A lower bound on $\|\Lambda_d\|_{\diamond}$}
\subsection{A lower bound on $\|\Lambda_d\|_{\diamond}$}

Let
\[
|\Omega_d\rangle
=
\frac{1}{\sqrt d}\sum_{j=0}^{d-1}|j\rangle\otimes |j\rangle
\]
be the maximally entangled state, and let
$
\Phi_d=|\Omega_d\rangle\langle \Omega_d|.
$
Then
\[
\|\Lambda_d\|_{\diamond}
\ge
\|(\Lambda_d\otimes \mathrm{id})(\Phi_d)\|_1.
\]
Let
\[
E_{ij}:=|i\rangle\langle j|,
\qquad
0\le i,j\le d-1.
\]
Expanding $\Phi_d$ in this basis gives
\[
\Phi_d
=
\frac{1}{d}\sum_{i,j=0}^{d-1}E_{ij}\otimes E_{ij}.
\]
Since
$
\Theta(E_{ij})=E_{ji},
$
we obtain
\[
(\Theta\otimes \mathrm{id})(\Phi_d)
=
\frac{1}{d}\sum_{i,j=0}^{d-1}E_{ji}\otimes E_{ij} = \frac{1}{d}F.
\]
Next we compute $(\Theta\circ \operatorname{Ad}_{U_d} \otimes \id)(\Phi_d)$:
\begin{align*}
& (\Theta\circ \operatorname{Ad}_{U_d} \otimes \id)(\Phi_d)\\
& = \frac{1}{d}\sum_{i,j=0}^{d-1}\big(U_d E_{ij} U_d^\dagger\big)^\top\otimes E_{ij}\\
& = \frac{1}{d}\sum_{i,j=0}^{d-1}\overline{U_d}\,E_{ji}\,U_d^\top\otimes E_{ij}\\
& = \frac{1}{d}
\big(\overline{U_d}\otimes \id \big)
\left(
\sum_{i,j=0}^{d-1}E_{ji}\otimes E_{ij}
\right)
\big(U_d^\top\otimes \id \big)\\
& = \frac{1}{d}
\big(\overline{U_d}\otimes \id \big)
F
\big(U_d^\top\otimes \id \big)\\
& =_{\text{(Lemma~\ref{lemma:swap_positions_operators})}} \frac{1}{d}F\big(U_d^\top\otimes \overline{U_d}\big).
\end{align*}
Subtracting this from $(\Theta\otimes \mathrm{id})(\Phi_d)=\frac{1}{d}F$, we find
$$
(\Lambda_d\otimes \mathrm{id})(\Phi_d) = \frac{1}{d}F\big(\id-U_d^\top\otimes \overline{U_d}\big).
$$
Thus
\begin{align*}
\|(\Lambda_d\otimes \mathrm{id})(\Phi_d)\|_1
& = \frac{1}{d}\left\|F\big(\id-U_d^\top\otimes \overline{U_d}\big)\right\|_1\\
& = \frac{1}{d}\left\| \id-U_d^\top\otimes \overline{U_d} \right\|_1
\end{align*}
and we get
\[
\|\Lambda_d\|_{\diamond}
\ge
\frac{1}{d}\left\|\id-U_d^\top\otimes \overline{U_d}\right\|_1.
\]
Meanwhile, it follows from the definition of $U_d$ that
$
\id - U_d^\top\otimes \overline{U_d}
$
is diagonal with diagonal entries
\[
1-\omega^{a-b},
\qquad
0\le a,b\le d-1.
\]
Therefore
$$
\|\Lambda_d\|_{\diamond}
\ge \frac{1}{d}\sum_{a,b=0}^{d-1}|1-\omega^{a-b}|
= \sum_{k=0}^{d-1}|1-\omega^k| = 2\cot\!\left(\frac{\pi}{2d}\right),
$$
which proves Eq.~\eqref{eq:Lambda-lower}.

% \subsection{Proof that $\|\mathrm{id}-\operatorname{Ad}_{U_d}\|_{\diamond}=2$}
\subsection{Proof that $\|\mathrm{id}-\operatorname{Ad}_{U_d}\|_{\diamond}=2$}

For unitary channels,
\begin{align*}
& \| \operatorname{Ad}_U - \operatorname{Ad}_V \|_\diamond\\
& = 2\sqrt{1 - \min\{|z|^2 : z \in \operatorname{conv}(\lambda_1, \dots , \lambda_d)\}},
\end{align*}
where $\lambda_1, \dots , \lambda_d$ are eigenvalues of $U^\dagger V$, and $\operatorname{conv}$ means convex hull \cite{haah2023query}. Since $\sum_{k=0}^{d-1} \omega^k = 0$, we get
$$
\| \id - \operatorname{Ad}_{U_d} \|_\diamond = 2,
$$
which proves Eq.~\eqref{eq:id-minus-unitary-exact}.

\bibliography{apssampv1}

\end{document}